\documentclass{article}

\usepackage{arxiv}

\usepackage[utf8]{inputenc} 
\usepackage[T1]{fontenc}    
\usepackage{hyperref}       
\usepackage{url}            
\usepackage{booktabs}       
\usepackage{amsfonts}       
\usepackage{nicefrac}       
\usepackage{microtype}      
\usepackage{lipsum}
\usepackage{graphicx}
\usepackage{bm}
\usepackage{mathtools}
\usepackage{amsmath}
\usepackage{amssymb}
\usepackage{amsthm}
\usepackage[numbers]{natbib}

\newtheorem{Theorem}{Theorem}[section]
\newtheorem{Definition}{Definition}[section]
\newtheorem{Corollary}{Corollary}[section]
\newtheorem{Lemma}{Lemma}[section]
\newtheorem{Proposition}{Proposition}[section]

\title{$\alpha$-Geodesical Skew Divergence}

\author{
 Masanari Kimura \\
  Graduate University for Advanced Studies, SOKENDAI \\
  \texttt{mkimura@ism.ac.jp} \\
   \And
 Hideitsu Hino \\
  Institute of Statistical Mathematics \\
  \texttt{hino@ism.ac.jp} \\
}

\begin{document}
\maketitle
\begin{abstract}
The asymmetric skew divergence smooths one of the distributions by mixing it, to a degree determined by the parameter $\lambda$, with the other distribution. Such divergence is an approximation of the KL-divergence that does not require the target distribution to be absolutely continuous with respect to the source distribution. In this paper, an information geometric generalization of the skew divergence called the  $\alpha$-geodesical skew divergence is proposed, and its properties are studied.
\end{abstract}

\section{Introduction}
Let $(\mathcal{X}, \mathcal{F}, \mu)$ be a measure space where $\mathcal{X}$ denotes the sample space, $\mathcal{F}$ the $\sigma$-algebra of measurable events, and $\mu$ a positive measure.
The set of the strictly positive probability measure $\mathcal{P}$ is defined as

\begin{equation}
    \mathcal{P} \coloneqq \Big\{f(x) > 0\ (\forall x\in\mathcal{X}),\ \text{and}\ \int_\mathcal{X} f(x)d\mu(x) = 1\Big\},
\end{equation}

and the set of nonnegative probability measure $\mathcal{P}_{+}$ is defined as

\begin{equation}
    \mathcal{P}_{+} \coloneqq \Big\{f(x) \geq 0\ (\forall x\in\mathcal{X}),\ \text{and}\ \int_\mathcal{X} f(x)d\mu(x) = 1\Big\}.
\end{equation}

Then a number of divergences that appear in statistics and information theory~\cite{deza2009encyclopedia,basseville2013divergence} are introduced.

\begin{Definition}{(Kullback–Leibler divergence~\cite{kullback1951information})}
The Kullback–Leibler divergence or KL-divergence $D_{KL}:\mathcal{P}_+\times\mathcal{P}\to[0,\infty]$ is defined between two Radon–Nikodym densities $p$ and $q$ of $\mu$-absolutely continuous probability measures by

\begin{equation}
    D_{KL}[p\|q] \coloneqq \int_\mathcal{X} p\ln\frac{p}{q}d\mu.
\end{equation}

\end{Definition}
KL-divergence is a measure of the difference between two probability distributions in statistics and information theory~\cite{sakamoto1986akaike,goldberger2003efficient,yu2013kl,solanki2006provably}.
This is also called the relative entropy, and is known not to satisfy the axiom of distance.
Since the KL-divergence is asymmetric, several symmetrizations have been proposed in the literature~\cite{lin1991divergence,menendez1997jensen,nielsen2019jensen}.
\begin{Definition}{(Jensen–Shannon divergence~\cite{lin1991divergence})}
The Jensen–Shannon divergence or JS-divergence $D_{JS}:\mathcal{P}\times\mathcal{P}\to[0,\infty)$ is defined between two Radon–Nikodym densities $p$ and $q$ of $\mu$-absolutely continuous probability measures by

\begin{align}
\notag
    D_{JS}[p\|q] &\coloneqq \frac{1}{2}\Biggl(D_{KL}\Big[p\Big\|\frac{p+q}{2}\Big] + D_{KL}\Big[q\Big\|\frac{p+q}{2}\Big]\Biggr) \\
    \notag
    &= \frac{1}{2}\int_\mathcal{X}\Biggl(p\ln\frac{2p}{p+q} + q\ln\frac{2q}{p+q}\Biggr) d\mu \\
    &= D_{JS}[q\| p].
\end{align}

\end{Definition}
The JS-divergence is a symmetrized and smoothed version of the KL-divergence, and it is bounded as

\begin{equation}
    0 \leq D_{JS}[p\|q] \leq \ln 2.
\end{equation}

This property contrasts with the fact that KL-divergence is unbounded.

\begin{Definition}{(Jeffreys divergence~\cite{jeffreys1946invariant})}
The Jeffreys divergence $D_J[p\|q]:\mathcal{P}\times\mathcal{P}\to[0,\infty]$ is defined between two Radon–Nikodym densities $p$ and $q$ of $\mu$-absolutely continuous probability measures by

\begin{equation}
    D_J[p\|q] \coloneqq D_{KL}[p\|q] + D_{KL}[q\|p].
\end{equation}

\end{Definition}

Such symmetrized KL-divergences have appeared in various literatures~\cite{chatzisavvas2005information,bigi2003using,wang2006groupwise,nishii2006image,bayarri2008generalization,nielsen2013jeffreys,nielsen2020generalization}.

For continuous distributions, the KL-divergence is known to have computational difficulty. 
To be more specific, if $q$ takes a small value relative to $p$, the value of $D_{KL}[p\|q]$ may diverge to infinity.
The simplest idea to avoid this is to use very small $\epsilon>0$ and modify $D_{KL}[p\|q]$ as follows:

\begin{equation*}
    D^+_{KL}[p\|q] \coloneqq \int_{\mathcal{X}} p\ln\frac{p}{q+\epsilon} d\mu.
\end{equation*}

However, such an extension is unnatural in the sense that $q+\epsilon$ no longer satisfies the condition for a probability measure: $\int_\mathcal{X} \epsilon + q(x) d\mu(x)\neq1$.
As a more natural way to stabilize KL-divergence, the following skew divergence have been proposed:
\begin{Definition}{(Skew divergence~\cite{lin1991divergence,lee1999measures})}
The skew divergence $D^{(\lambda)}_S[p\|q]:\mathcal{P}\times\mathcal{P}\to[0,\infty]$ is defined between two Radon–Nikodym densities $p$ and $q$ of $\mu$-absolutely continuous probability measures by

\begin{align}
\notag
    D^{(\lambda)}_S[p\|q] &\coloneqq D_{KL}[p\|(1-\lambda)p + \lambda q] \\
    &= \int_\mathcal{X} p\ln\frac{p}{(1-\lambda)p + \lambda q} d\mu,
\end{align}

where $\lambda\in[0,1]$.
\end{Definition}
Skew divergences have been
experimentally shown to perform better in applications such as natural language processing~\cite{lee2001effectiveness,xiao2019dual}, image recognition~\cite{carvalho2014skew,revathi2014cotton} and graph analysis~\cite{ahmed2011network,hughes2007lexical}.
In addition, there is research on quantum generalization of skew divergence~\cite{audenaert2014quantum}.

The main contributions of this paper are summarized as follows:
\begin{itemize}
    \item Several symmetrized divergences or skew divergences are generalized from an information geometry perspective. 
    \item It is proved that the natural skew divergence for the exponential family is equivalent to the scaled KL-divergence.
    \item Several properties of geometrically generalized skew divergence are proved. Specifically, the functional space associated with the proposed divergence is shown to be a Banach space.
\end{itemize}
Implementation of the proposed divergence is available on GitHub\footnote{\url{https://github.com/nocotan/geodesical_skew_divergence}}. 

\section{\texorpdfstring{$\alpha$}--Geodesical Skew Divergence}

The skew divergence is generalized based on the following function.
\begin{Definition}{($f$-interpolation)}
\label{def:f_interpolation}
For any $a,b,\in\mathbb{R}$,  $\lambda\in[0,1]$ and  $\alpha\in\mathbb{R}$,  $f$-interpolation is defined as

\begin{equation}
    \label{eq:f_interpolation}
    m_f^{(\lambda,\alpha)}(a,b) = f^{-1}_\alpha\Big( (1-\lambda) f_{\alpha}(a) + \lambda f_{\alpha}(b) \Big),
\end{equation}

where

\begin{equation}
    f_\alpha(x) = \begin{cases}
    x^{\frac{1-\alpha}{2}} & (\alpha\neq 1) \\
    \ln x & (\alpha = 1)
    \end{cases}
\end{equation}

is the function that defines the $f$-mean~\cite{hardy1952inequalities}.
\end{Definition}
The $f$-mean function satisfies

\begin{align*}
    \lim_{\alpha\to\infty}f_\alpha(x) &= \begin{cases}
    \infty & (|x| < 1), \\
    1 & (|x| = 1), \\
    0 & (|x| > 1),
    \end{cases} \\
   \lim_{\alpha\to -\infty}f_\alpha(x) &=  \begin{cases}
   0 & (|x| < 1), \\
   1 & (|x| = 1), \\
   \infty & (|x| > 1).
   \end{cases}
\end{align*}

It is easy to see that this family includes various known weighted means including the $e$-mixture and $m$-mixture for $\alpha=\pm 1$ in the literature of information geometry~\cite{Amari2016-pi}:

\begin{align}
    (\alpha = 1)\ \ &\ m_f^{(\lambda,1)}(a,b) = \exp\{(1-\lambda)\ln a + \lambda \ln b\} \nonumber \\
    (\alpha = -1)\ \ &\ m_f^{(\lambda, -1)}(a,b) = (1-\lambda)a + \lambda b \nonumber \\
    (\alpha = 0)\ \ &\ m_f^{(\lambda, 0)}(a,b) = \Big((1-\lambda)\sqrt{a} + \lambda\sqrt{b}\Big)^2 \nonumber \\
    (\alpha = 3)\ \ &\ m_f^{(\lambda, 3)}(a,b) = \frac{1}{(1-\lambda)\frac{1}{a} + \lambda\frac{1}{b}} \nonumber \\
    (\alpha=\infty)\ \ &\ m_f^{(\lambda, \infty)}(a,b) = \min\{a, b\} \nonumber \\
    (\alpha=-\infty)\ \ &\ m_f^{(\lambda, -\infty)}(a,b) = \max\{a, b\} \nonumber
\end{align}

\if
For any $\bm{u},\bm{v}\in\mathbb{R}^d\ (d>0)$, we write
\begin{align*}
    \bm{w} &= m_f^{(\lambda, \alpha)}(\bm{u}, \bm{v}), \\ 
    \mbox{where}
    \quad
    \bm{w}_i &= m_f^{(\lambda, \alpha)}(\bm{u}_i, \bm{v}_i).
\end{align*}
\fi

The inverse function $f^{-1}_\alpha$ is convex when $\alpha \in [-1,1]$, and concave when $\alpha \in (-\infty, -1] \cup (1, \infty)$.
It is worth noting that the $f$-interpolation is a special case of the Kolmogorov-Nagumo average~\cite{kolmogorov1930notion,nagumo1930klasse,nielsen2014generalized} when $\alpha$ is restricted in the interval $[-1,1]$.

In order to consider the geometric meaning of this function, the notion of the statistical manifold is introduced.
\subsection{Statistical Manifold}
Let

\begin{equation}
    \mathcal{S} = \{p_{\bm{\xi}} = p(\bm{x};\bm{\xi})  \in \mathcal{P}| \bm{\xi}=(\xi^1,\dots,\xi^n)\in\Xi\}
\end{equation}

be a family of probability distribution on $\mathcal{X}$, where each element $p_{\bm{\xi}}$ is parameterized by $n$ real-valued variables $\bm{\xi}=(\xi^1,\dots,\xi^n)\in\Xi\subset\mathbb{R}^n$.
The set $\mathcal{S}$ is called a statistical model and is a subset of $\mathcal{P}$.
We also denote $(\mathcal{S}, g_{ij})$ as a statistical model equipped with the Riemannian metric $g_{ij}$.
In particular, let $g_{ij}$ be the Fisher-Rao metric, which is the Riemannian metric induced from the Fisher information matrix~\cite{amari2012differential}.

In the rest of this paper, the abbreviations
\begin{align*}
    \partial_i &= \partial_{\xi^i} = \frac{\partial}{\partial\xi^i}, \\
    \ell &= \ell_{\bm{x}}(\bm{\xi}) = \ln p_{\bm{\xi}}(\bm{x})
\end{align*}
are used.

\begin{Definition}{(Christoffel symbols)}
Let $g_{ij}$ be a Riemannian metric, particularly the Fisher information matrix, then the Christoffel symbols are given by

\begin{equation}
    \Gamma_{ij,k} = \frac{1}{2}\Big(\partial_i g_{jk} + \partial_j g_{ik} - \partial_k g_{ij}\Big), \quad 
    i,j,k = 1, \dots, n.
\end{equation}

\end{Definition}

\begin{Definition}{(Levi-Civita connection)}
Let $g$ be a Fisher-Riemannian metric on $\mathcal{S}$ which is a 2-covariant tensor defined locally by

\begin{equation*}
    g(X_{\bm{\xi}}, Y_{\bm{\xi}}) = \sum^n_{i,j=1}g_{ij}(\bm{\xi})a^i(\bm{\xi})b^j(\bm{\xi}),
\end{equation*}

where $X_{\bm{\xi}}=\sum^n_{i=1}a^i(\bm{\xi})\partial_i p_{\bm{\xi}}$ and $Y_{\bm{\xi}}=\sum^n_{i=1}b^i(\bm{\xi})\partial_i p_{\bm{\xi}}$ are vector fields in the 0-representation on $\mathcal{S}$.
Then, its associated Levi-Civita connection $\nabla^{(0)}$ is defined by

\begin{equation}
    g(\nabla^{(0)}_{\partial_i}\partial_j, \partial_k) = \Gamma_{ij,k}.
\end{equation}

\end{Definition}
The fact that $\nabla^{(0)}$ is metrical connection can be written locally as

\begin{equation}
    \partial_k g_{ij} = \Gamma_{ki,j} + \Gamma_{kj,i}.
\end{equation}

It is worth noting that the superscript $\alpha$ of $\nabla^{(\alpha)}$ corresponds to a parameter of the connection.
Based on the above definitions several connections parameterized by the parameter $\alpha$ are introduced.
The case $\alpha=0$ corresponds to the Levi-Civita connection induced by the Fisher metric.

\begin{Definition}{($\nabla^{(1)}$-connection)}
Let $g$ be the Fisher-Riemannian metric on $\mathcal{S}$ which is a 2-covariant tensor.
Then, the $\nabla^{(1)}$-connection is defined by

\begin{equation}
    g(\nabla^{(1)}_{\partial_i}\partial_j, \partial_k) = \mathbb{E}_{\bm{\xi}}[\partial_i\partial_j\ell \partial_k\ell].
\end{equation}

It can also be expressed equivalently by explicitly writting as the Christoffel coefficients

\begin{equation}
    \Gamma^{(1)}_{ij,k}(\bm{\xi}) = \mathbb{E}_{\bm{\xi}}[\partial_i\partial_j\ell\partial_k\ell].
\end{equation}

\end{Definition}

\begin{Definition}{($\nabla^{(-1)}$-connection)}
Let $g$ be the Fisher-Riemannian metric on $\mathcal{S}$ which is a 2-covariant tensor.
Then, the $\nabla^{(-1)}$-connection is defined by

\begin{equation}
    g(\nabla^{(-1)}_{\partial_i}\partial_j, \partial_k) = \Gamma^{(-1)}_{ij,k}(\bm{\xi}) = \mathbb{E}_{\bm{\xi}}[(\partial_i\partial_j\ell + \partial_i\ell\partial_j\ell)\partial_k\ell].
\end{equation}

\end{Definition}
In the following the $\nabla$-flatness is considered with respect to the corresponding coordinates system.
More details can be found in~\cite{Amari2016-pi}.
\begin{Proposition}
The exponential family is $\nabla^{(1)}$-flat.
\end{Proposition}
\begin{Proposition}
The exponential family is $\nabla^{(-1)}$-flat if and only if it is $\nabla^{(0)}$-flat.
\end{Proposition}
\begin{Proposition}
The mixture family is $\nabla^{(-1)}$-flat.
\end{Proposition}
\begin{Proposition}
The mixture family is $\nabla^{(1)}$-flat if and only if it is $\nabla^{(0)}$-flat.
\end{Proposition}
\begin{Proposition}
The relation between the foregoing three connections is given by

\begin{equation}
    \nabla^{(0)} = \frac{1}{2}\Big(\nabla^{(-1)} + \nabla^{(1)}\Big).
\end{equation}

\end{Proposition}
\begin{proof}
It suffices to show

\begin{equation*}
    \Gamma^{(0)}_{ij,k} = \frac{1}{2}\Big(\Gamma^{(-1)}_{ij,k} + \Gamma^{(1)}_{ij,k}\Big).
\end{equation*}

From the definitions of $\Gamma^{(-1)}$ and $\Gamma^{(1)}$,

\begin{align*}
    \Gamma^{(-1)}_{ij,k} + \Gamma^{(1)}_{ij,k} &= \mathbb{E}_{\bm{\xi}}[(\partial_i\partial_j\ell + \partial_i\ell\partial_j\ell)\partial_k\ell] + \mathbb{E}_{\bm{\xi}}[\partial_i\partial_j\ell\partial_k\ell] \\
    &= \mathbb{E}_{\bm{\xi}}[(2\partial_i\partial_j\ell + \partial_i\ell\partial_j\ell)\partial_k\ell] \\
    &= 2\mathbb{E}_{\bm{\xi}}\Big[(\partial_i\partial_j\ell + \frac{1}{2}\partial_i\ell\partial_j\ell)\partial_k\ell\Big] \\
    &= 2\Gamma^{(0)}_{ij,k},
\end{align*}

which proves the proposition.
\end{proof}

The connections $\nabla^{(-1)}$ and $\nabla^{(1)}$ are two special connections on $\mathcal{S}$ with respect to the mixture family and the exponential family, respectively.
Moreover, they are related by the duality condition, and the following 1-parameter family of connections are defined.
\begin{Definition}{($\nabla^{(\alpha)}$-connection)}
For $\alpha\in\mathbb{R}$, the $\nabla^{(\alpha)}$-connection on the statistical model $\mathcal{S}$ is defined as

\begin{equation}
    \nabla^{(\alpha)} = \frac{1+\alpha}{2}\nabla^{(1)} + \frac{1-\alpha}{2}\nabla^{(-1)}.
\end{equation}

\end{Definition}

\begin{Proposition}
The components $\Gamma^{(\alpha)}_{ij,k}$ can be written as

\begin{equation}
    \Gamma^{(\alpha)}_{ij,k} = \mathbb{E}_{\bm{\xi}}\Biggl[\Big(\partial_i\partial_j\ell + \frac{1-\alpha}{2}\partial_i\ell\partial_j\ell\Big)\partial_k\ell\Biggr].
\end{equation}

\end{Proposition}

The $\alpha$-coordinate system associated with the  $\nabla^{(\alpha)}$-connection is endowed with the $\alpha$-geodesic which is a straight line on the corresponding coordinates system.
Then, we introduce some relevant notions.

\begin{Definition}{($\alpha$-divergence~\cite{Amari1985-mi})}
Let $\alpha$ be a real parameter.
The $\alpha$-divergence between two probability vectors $\bm{p}$ and $\bm{q}$ is defined as

\begin{equation}
    D_\alpha[\bm{p}\| \bm{q}] = \frac{4}{1-\alpha^2}\Big(1 - \sum_i p_i^{\frac{1-\alpha}{2}}q_i^{\frac{1+\alpha}{2}}\Big).
\end{equation}

\end{Definition}
The KL-divergence, which is a special case with $\alpha=1$, induces the linear connection $\nabla^{(1)}$ as follows.
\begin{Proposition}
The diagonal part of the third mixed derivatives of the KL-divergence is the negative of the Christoffel symbol:

\begin{equation}
    -\partial_{\bm{\xi}^i}\partial_{\bm{\xi}^j}\partial_{\bm{\xi}^k_0}D_{KL}[p_{\bm{\xi}_0}\| p_{\bm{\xi}}]\Big|_{\bm{\xi}=\bm{\xi}_0} = \Gamma^{(1)}_{ij,k}(\bm{\xi}_0).
\end{equation}

\end{Proposition}
\begin{proof}
The second derivative in the argument $\bm{\xi}$ is given by

\begin{equation*}
    \partial_{\bm{\xi}^i}\partial_{\bm{\xi}^j}D_{KL}[p_{\bm{\xi}_0}\|p_{\bm{\xi}}] = -\int_\mathcal{X} p_{\bm{\xi}_0}(\bm{x})\partial_{\bm{\xi}^i}\partial_{\bm{\xi}^j}\ell_{\bm{x}}(\bm{\xi})d\bm{x},
\end{equation*}

and differentiating it with respect to $\bm{\xi}^k_0$ yields

\begin{align*}
    -\partial_{\bm{\xi}^i}\partial_{\bm{\xi}^j}\partial_{\bm{\xi}^k_0}D_{KL}[p_{\bm{\xi}_0}\| p_{\bm{\xi}}] &= \partial_{\bm{\xi}^k_0}\int_\mathcal{X} p_{\bm{\xi}_0}(\bm{x})\partial_{\bm{\xi}^i}\partial_{\bm{\xi}^j}\ell_{\bm{x}}(\bm{\xi})d\bm{x} \\
    &=  \int_\mathcal{X} p_{\bm{\xi}_0}(\bm{x})\partial_{\bm{\xi}^i}\partial_{\bm{\xi}^j}\ell_{\bm{x}}(\bm{\xi})\partial_{\bm{\xi}^k_0}\ell_{\bm{x}}(\bm{\xi})d\bm{x}.
\end{align*}

Then, considering the diagonal part, one yields

\begin{align*}
    -\partial_{\bm{\xi}^i}\partial_{\bm{\xi}^j}\partial_{\bm{\xi}^k_0}D_{KL}[p_{\bm{\xi}_0}\| p_{\bm{\xi}}]\Big|_{\bm{\xi}=\bm{\xi}_0} &= \mathbb{E}_{\bm{\xi}_0}[\partial_i\partial_j\ell(\bm{\xi})\partial_k\ell(\bm{\xi})]\\
    &= \Gamma^{(1)}_{ij,k}(\bm{\xi}_0).
\end{align*}
\end{proof}

More generally, the $\alpha$-divergence with $\alpha\in\mathbb{R}$ induces the $\nabla^{(\alpha)}$-connection.

\begin{Definition}{($\alpha$-representation~\cite{Amari2009-fz})}
For some positive measure $m_i^{\frac{1-\alpha}{2}}$, the coordinate system $\bm{\theta}=(\theta^i)$ derived from the $\alpha$-divergence is

\begin{equation}
    \theta^i = m_i^{\frac{1-\alpha}{2}} = f_\alpha(m_i)
\end{equation}

and $\theta^i$ is called the $\alpha$-representation of a positive measure $m_i^{\frac{1-\alpha}{2}}$.
\end{Definition}
\begin{Definition}{($\alpha$-geodesic~\cite{Amari2016-pi})}
\label{def:alpha_geodesic}
The $\alpha$-geodesic connecting two probability vectors $p(\bm{x})$ and $q(\bm{x})$ is defined as

\begin{equation}
    r_i(t) = c(t)f^{-1}_{\alpha}\Big\{(1-t)f_{\alpha}(p(x_i)) + t f_\alpha(q(x_i))\Big\}, \quad t \in [0,1]
\end{equation}

where $c(t)$ is determined as

\begin{equation}
    c(t) = \frac{1}{\sum^n_{i=1}r_i(t)}.
\end{equation}

\end{Definition}

It is known that the appropriate reparameterizations for the parameter $t$ is necessary for a rigorous discussion in the space of probability measures~\cite{Ay2017,Morozova1991}.
However, as mentioned in the literature~\cite{Ay2017}, an explicit expression for the reparametrizations $\tau_{p,a}$ and $\tau_{p,q}$ is unknown.
A similar discussion has been made in the derivation of the $\phi_{\beta}$-path~\cite{Eguchi2015}, where it is mentioned that the normalizing factor is unknown in general.
Furthermore, the $f$-mean is not convex depending on the $\alpha$.
For these reasons, it is generally difficult to discuss $\alpha$-geodesics in probability measures by normalization or reparameterization, and to avoid unnecessary complexity, the parameter $t$ is assumed to be appropriately reparameterized.

Let $\psi_\alpha(\bm{\theta}) = \frac{1-\alpha}{2}\sum_{i=1}^{n} m_i$. Then, the dual coordinate system $\bm{\eta}$ is given by $\bm{\eta} = \nabla\psi_\alpha(\bm{\theta})$ as

\begin{equation}
    \eta_i = (\theta^i)^{\frac{1+\alpha}{1-\alpha}} = f_{-\alpha}(m_i).
\end{equation}

Hence it is the ($-\alpha$)-representation of $m_i$.

\subsection{Generalization of Skew Divergences}

From Definition~\ref{def:alpha_geodesic}, the $f$-interpoloation is considered as an unnormalized version of the $\alpha$-geodesic.
Using the notion of geodesics, skew divergence is generalized in terms of information geometry as follows.

\begin{Definition}{($\alpha$-Geodesical Skew Divergence)}
The $\alpha$-geodesical skew divergence $D_{GS}^{(\alpha,\lambda)}:\mathcal{P}\times\mathcal{P}\to[0,\infty]$ is defined between two Radon–Nikodym densities $p$ and $q$ of $\mu$-absolutely continuous probability measures by:

\begin{align}
    D_{GS}^{(\alpha,\lambda)}\Big[p\|q\Big] &\coloneqq D_{KL}\Big[p\|m_f^{(\lambda,\alpha)}(p,q)\Big] \nonumber \\
    &= \int_\mathcal{X} p\ln\frac{p}{m_f^{(\lambda,\alpha)}(p,q)} d\mu,
\end{align}

where $\alpha\in\mathbb{R}$ and $\lambda\in[0,1]$.
\end{Definition}
Some special cases of $\alpha$-geodesical skew divergence are listed below:

\begin{align*}
(\forall\alpha\in\mathbb{R}, \lambda = 1)\ & D_{GS}^{(\alpha,1)} [p\|q] = D_{KL} [p\|q] \\
(\forall\alpha\in\mathbb{R}, \lambda = 0)\ & D_{GS}^{(\alpha,0)} [p\|q] = D_{KL}[p\|p] = 0 \\
(\alpha=1,\forall\lambda\in[0,1])\ & D_{GS}^{(1, \lambda)} [p\|q] = \lambda D_{KL} [p\|q]\ \ (\text{scaled KL-divergence}) \\
(\alpha=-1,\forall\lambda\in[0,1])\ & D_{GS}^{(-1, \lambda)} [p\|q] = D_S^{(\lambda)} [p\|q]\ \ (\text{skew divergence}) \\
(\alpha=0, \forall\lambda\in[0,1])\ & D_{GS}^{(0, \lambda)} [p\|q] = \int_\mathcal{X} p\ln\frac{p}{\{(1-\lambda)\sqrt{p} + \lambda\sqrt{q}\}^2}d\mu \\
(\alpha=3, \forall\lambda\in[0,1])\ & D_{GS}^{(3, \lambda)} [p\|q] = D_S^{(\lambda)} [p\|q] + H(p) + H(q) \\
(\alpha=\infty, \forall{\lambda}\in[0,1])\ & D_{GS}^{(\infty,\lambda)} [p\|q] = \int_\mathcal{X} p\ln\frac{p}{\min\{p, q\}} d\mu \\
(\alpha=-\infty, \forall{\lambda}\in[0,1])\ & D_{GS}^{(-\infty,\lambda)} [p\|q] = \int_\mathcal{X} p\ln\frac{p}{\max\{p, q\}} d\mu
\end{align*}
Also, $\alpha$-geodesical skew divergence is a special form of the generalized skew K-divergence~\cite{nielsen2019jensen,e23040464}, which is a family of abstract means-based divergences.
In this paper, the skew K-divergence touched upon in~\cite{nielsen2019jensen} is characterized in terms of $\alpha$-geodesic on positive measures, and its geometric and functional analytic properties are investigated.
When the Kolmogorov-Nagumo average (i.e., when the function $f^{-1}$ in Eq.~\eqref{eq:f_interpolation} is a strictly monotone convex function) the geodesic has been shown to be well-defined~\cite{Eguchi2015}.

\subsection{Symmetrization of \texorpdfstring{$\alpha$}--Geodesical Skew Divergence}
It is easy to symmetrize the $\alpha$-geodesical skew divergence as follows.
\begin{Definition}{(Symmetrized $\alpha$-Geodesical Skew Divergence)}
The symmetrized $\alpha$-geodesical skew divergence $\bar{D}^{(\alpha,\lambda)}_{GS}:\mathcal{P}\times\mathcal{P}\to[0,\infty]$ is defined between two Radon–Nikodym densities $p$ and $q$ of $\mu$-absolutely continuous probability measures by:

\begin{align}
    \bar{D}^{(\alpha,\lambda)}_{GS}[p\|q] &\coloneqq \frac{1}{2}\Biggl(D_{GS}^{(\alpha,\lambda)}[p\|q] + D_{GS}^{(\alpha,\lambda)}[q\|p] \Biggr),
\end{align}

where $\alpha\in\mathbb{R}$ and $\lambda\in[0,1]$.
\end{Definition}
It is seen that $\bar{D}^{(\alpha,\lambda)}_{GS}[p\|q]$ includes several symmetrized divergences.

\begin{align*}
     \bar{D}_{GS}^{(\alpha,1)}[p\|q] &= \frac{1}{2}\Biggl(D_{KL}[p\|q] + D_{KL}[q\|p]\Biggr),\ \ (\text{half of Jeffreys divergence}) \\
    \bar{D}_{GS}^{(-1,\frac{1}{2})}[p\|q] &= \frac{1}{2}\Biggl(D_{KL}\Big[p\|\frac{p+q}{2}\Big] + D_{KL}\Big[q\|\frac{p+q}{2}\Big]\Biggr),\ \ (\text{JS-divergence}) \\
    \bar{D}_{GS}^{(-1,\lambda)}[p\|q] &= \frac{1}{2}\Biggl(D_{KL}\Big[p\|(1-\lambda)p + \lambda q\Big] + D_{KL}\Big[q\|(1-\lambda)q + \lambda p\Big]\Biggr).
\end{align*}

The last one is the $\lambda$-JS-divergence~\cite{nielsen2010family}, which is a generalization of the JS-divergence.

\section{Properties of \texorpdfstring{$\alpha$}--Geodesical Skew Divergence}
In this section, the properties of the $\alpha$-geodesical skew divergence are studied.

\begin{Proposition}{(Non-negativity of the $\alpha$-geodesical skew divergence)}
For $\alpha\geq -1$ and $\lambda\in[0,1]$, the $\alpha$-geodesical skew divergence $D_{GS}^{(\alpha,\lambda)}[p\|q]$ satisfies the following inequality:

\begin{equation}
    D_{GS}^{(\alpha,\lambda)}[p\|q] \geq 0.
\end{equation}

\begin{proof}
When $\lambda$ is fixed, the $f$-interpolation has the following inverse monotonicity with respect to $\alpha$:

\begin{equation}
    m_f^{(\lambda,\alpha)}(p,q) \geq m_f^{(\lambda,\alpha')}(p,q),\ (\alpha \leq \alpha'). \label{eq:f_interpolation_inverse}
\end{equation}

From Gibbs' inequality~\cite{cover1999elements} and Eq.~\eqref{eq:f_interpolation_inverse}, one obtains

\begin{align*}
    D_{GS}^{(\alpha,\lambda)}[p\|q] &= \int_\mathcal{X} p\ln\frac{p}{m_f^{(\alpha,\lambda)}(p,q)} d\mu \nonumber \\
    &\geq \Big(\int_\mathcal{X} p d\mu\Big)\ln\frac{p}{m_f^{(\alpha,\lambda)}(p,q)} \nonumber \\
    &\geq 1\cdot \ln 1 = 0. \nonumber
\end{align*}
\end{proof}
\end{Proposition}

\begin{Proposition}{(Asymmetry of the $\alpha$-geodesical skew divergence)}
$\alpha$-Geodesical skew divergence is not symmetric in general:

\begin{equation}
    D_{GS}^{(\alpha, \lambda)} [p| q] \neq D_{GS}^{(\alpha, \lambda)}[q\| p].
\end{equation}

\end{Proposition}
\begin{proof}
For example, if $\lambda=1$, then $\forall\alpha\in\mathbb{R}$, it holds that

\begin{align}
    D_{GS}^{(\alpha,1)}[p\|q] -  D_{GS}^{(\alpha,1)}[q\|p] &= D_{KL}[p\|q] - D_{KL}[q\|p], \nonumber
\end{align}

and the asymmetry of the KL-divergence results in an asymmetry of the geodesic skew divergence.
\end{proof}

When a function $f(x)$ of $x\in [0,1]$ satisfies $f(x) = f(1-x)$, it is referred to be centrosymmetric.
\begin{Proposition}{(Non-centrosymmetricicy of the $\alpha$-geodesical skew divergence with respect to $\lambda$)}
$\alpha$-Geodesical skew divergence is not centrosymmetric  in general with respect to the parameter $\lambda \in [0,1]$:

\begin{equation}
    D_{GS}^{(\alpha, \lambda)} [p\|q] \neq D_{GS}^{(\alpha, 1-\lambda)} [p\| q].
\end{equation}

\end{Proposition}
\begin{proof}
For example, if $\lambda=1$, then $\forall\alpha\in\mathbb{R}$, we have

\begin{align}
    D_{GS}^{(\alpha,\lambda)}[p\|q] - D_{GS}^{(\alpha,1-\lambda)}[p\|q] &=  D_{GS}^{(\alpha,1)}[p\|q] - D_{GS}^{(\alpha,0)}[p\|q] \\
    &= \int_\mathcal{X} p\ln\frac{p}{q} - \int_\mathcal{X} p\ln\frac{p}{p} \nonumber \\
    &= \int_\mathcal{X} p\ln\frac{p}{q} \geq 0. \nonumber
\end{align}
\end{proof}

\begin{figure}[t]
    \centering
    \includegraphics[scale=0.45]{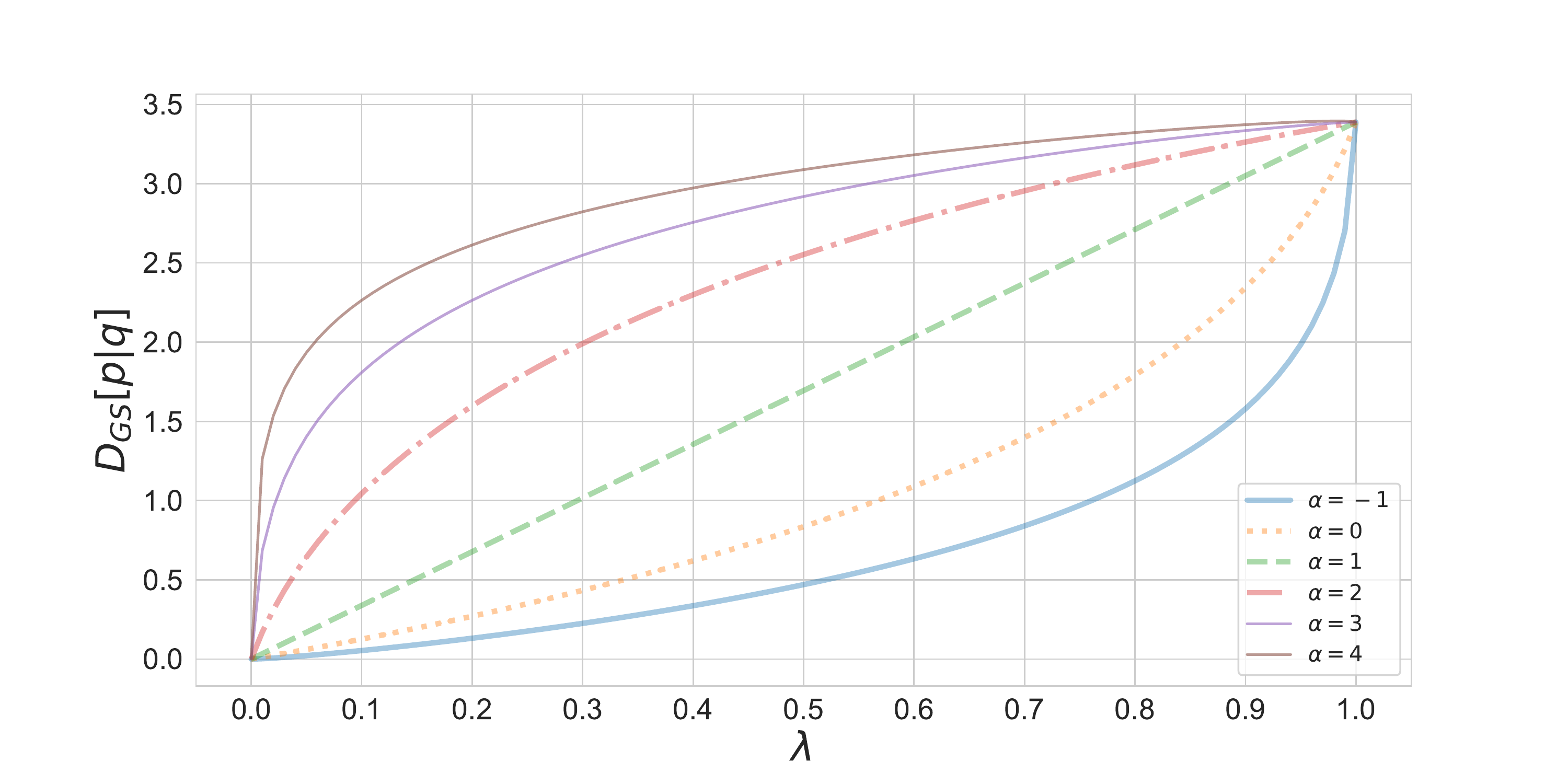}
    \caption{Monotonicity of the $\alpha$-geodesical skew divergence with respect to $\alpha$. The $\alpha$-geodesical skew divergence between the binomial distributions $p=B(10,0.3)$ and $q=B(10,0.7)$ has been calculated.}
    \label{fig:gsd_monotonicity}
\end{figure}

\begin{Proposition}{(Monotonicity of the $\alpha$-geodesical skew divergence with respect to $\alpha$)}
$\alpha$-Geodesical skew divergence satisfies the following inequality for all $\alpha\in\mathbb{R},\lambda\in[0,1]$.

\begin{equation*}
    D_{GS}^{(\alpha,\lambda)}[p\|q] \geq D_{GS}^{(\alpha',\lambda)}[p\|q],\ (\alpha\geq\alpha').
\end{equation*}

\end{Proposition}
\begin{proof}
Obvious from the inverse monotonicity of the $f$-interpolation~\eqref{eq:f_interpolation_inverse} and the monotonicity of the logarithmic function.
\end{proof}
Figure~\ref{fig:gsd_monotonicity} shows the monotonicity of the geodesic skew divergence with respect to $\alpha$.
In this figure, divergence is calculated between two binomial distributions.

\begin{Proposition}{(Subadditivity of the $\alpha$-geodesical skew divergence with respect to $\alpha$)}
$\alpha$-Geodesical skew divergence satisfies the following inequality for all $\alpha,\beta\in\mathbb{R},\lambda\in[0,1]$
\end{Proposition}

\begin{equation*}
    D_{GS}^{(\alpha+\beta,\lambda)}[p\|q] \leq D_{GS}^{(\alpha,\lambda)}[p\|q] + D_{GS}^{(\beta,\lambda)}[p\|q].
\end{equation*}

\begin{proof}
For some $\alpha$ and $\lambda$, $m_f^{(\lambda,\alpha)}$ takes the form of the Kolmogorov mean~\cite{kolmogorov1930notion}, and obvious from its continuity, monotonicity and self-distributivity.
\end{proof}

\begin{Proposition}{(Continuity of the $\alpha$-geodesical skew divergence with respect to $\alpha$ and $\lambda$)}
$\alpha$-Geodesical skew divergence has the continuity property.
\end{Proposition}
\begin{proof}
We can prove from the continuity of the KL-divergence and the Kolmogorov mean.
\end{proof}
Figure~\ref{fig:gsd_surface} shows the continuity of the geodesic skew divergence with respect to $\alpha$ and $\lambda$.
Both of source and target distributions are binomial distributions. From this figure, it can be seen that the divergence changes smoothly as the parameters change.

\begin{figure}[t]
    \centering
    \includegraphics[scale=0.45]{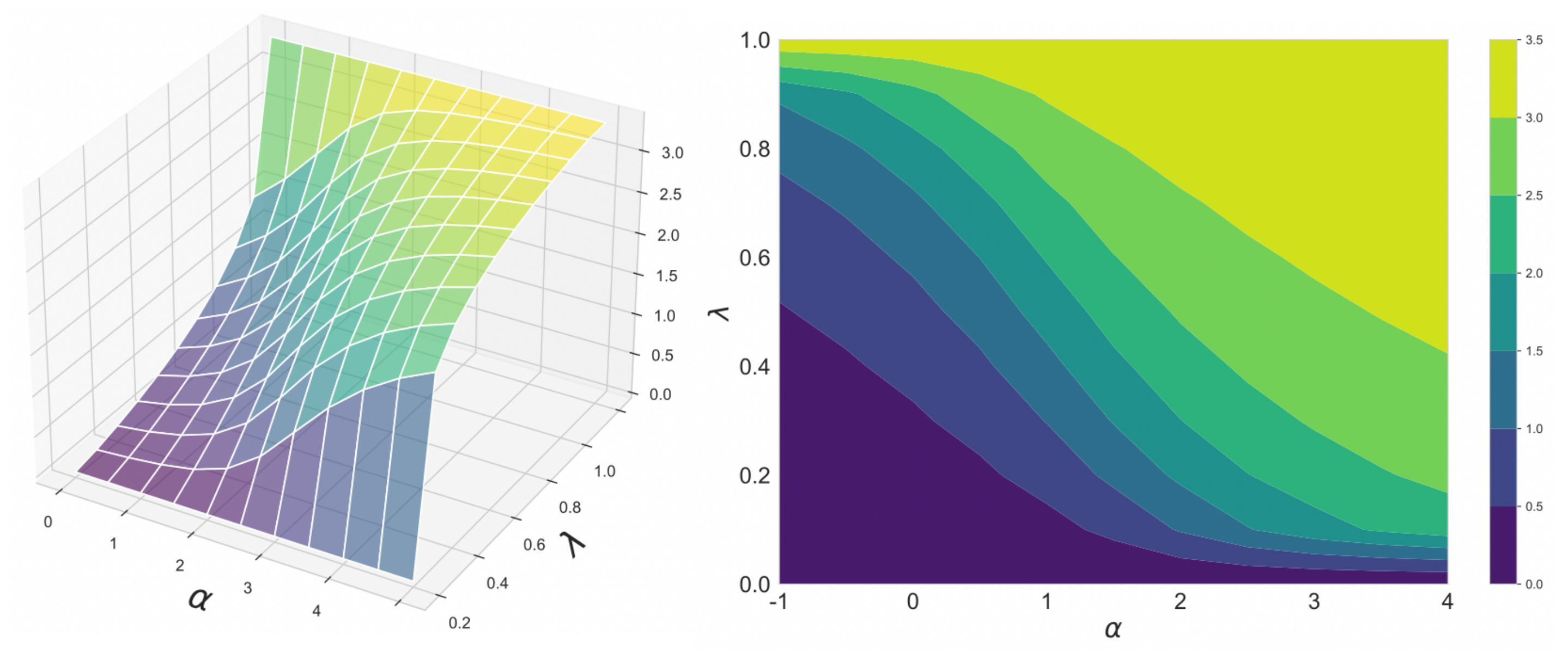}
    \caption{Continuity of the $\alpha$-geodesical skew divergence with respect to $\alpha$ and $\lambda$. The $\alpha$-geodesical skew divergence between the binomial distributions $p=B(10,0.3)$ and $q=B(10,0.7)$ has been calculated.}
    \label{fig:gsd_surface}
\end{figure}

\begin{Lemma}
\label{lem:gsd_alpha_infty}
Suppose $\alpha\to\infty$. Then,

\begin{equation}
    \lim_{\alpha\to\infty}D_{GS}^{(\alpha,\lambda)}[p\|q] = \int_\mathcal{X} p\ln\frac{p}{\min{\{p,q\}}} d\mu
\end{equation}

holds for all $\lambda\in[0,1]$.
\end{Lemma}
\begin{proof}
Let $u = \frac{1-\alpha}{2}$. Then $\lim_{\alpha\to\infty} u = -\infty$.
Assuming $p_0 \leq p_1$, it holds that

\begin{align*}
\lim_{\alpha\to\infty} m_f^{(\lambda,\alpha)}(p_0,p_1) &= \lim_{u\to -\infty}\Biggl((1-\lambda)p^u_0 + \lambda p^u_1\Biggr)^{\frac{1}{u}} \\
&= p_0\lim_{u\to -\infty}\Biggl((1-\lambda) + \lambda \Big(\frac{p_1}{p_0}\Big)^u\Biggr)^{\frac{1}{u}} \\
&= p_0 = \min{\{p_0, p_1\}}.
\end{align*}

Then, the following equality

\begin{align*}
    \lim_{\alpha\to\infty}D_{GS}^{(\alpha,\lambda)}[p\|q] &= \int_\mathcal{X} p\ln \frac{p}{\lim_{\alpha\to\infty} m_f^{(\lambda,\alpha)}(p_0,p_1)}d\mu \\
    &= \int_\mathcal{X} p\ln\frac{p}{\min{\{p,q\}}} d\mu
\end{align*}

holds.
\end{proof}

\begin{Lemma}
\label{lem:gsd_alpha_minus_infty}
Suppose $\alpha\to -\infty$. Then,

\begin{equation}
    \lim_{\alpha\to\infty}D_{GS}^{(\alpha,\lambda)}[p\|q] = \int_\mathcal{X} p\ln \frac{p}{\max{\{p,q\}}}d\mu
\end{equation}

holds for all $\lambda\in[0,1]$.
\end{Lemma}
\begin{proof}
Let $u = \frac{1-\alpha}{2}$. Then  $\lim_{\alpha\to -\infty} u = \infty$.
Assuming $p_0 \leq p_1$, it holds that

\begin{align*}
\lim_{\alpha\to\infty} m_f^{(\lambda,\alpha)}(p_0,p_1) &= \lim_{u\to -\infty}\Biggl((1-\lambda)p^u_0 + \lambda p^u_1\Biggr)^{\frac{1}{u}} \\
&= p_1\lim_{u\to -\infty}\Biggl((1-\lambda)\Big(\frac{p_0}{p_1}\Big)^u + \lambda \Biggr)^{\frac{1}{u}} \\
&= p_1 = \max{\{p_0, p_1\}}.
\end{align*}

Then, the following equality

\begin{align*}
    \lim_{\alpha\to -\infty}D_{GS}^{(\alpha,\lambda)}[p\|q] &= \int_\mathcal{X} p\ln \frac{p}{\lim_{\alpha\to -\infty} m_f^{(\lambda,\alpha)}(p_0,p_1)}d\mu \\
    &= \int_\mathcal{X} p\ln \frac{p}{\max{\{p,q\}}} d\mu
\end{align*}

holds.
\end{proof}

\begin{Proposition}{(Lower bound of the $\alpha$-geodesical skew divergence)}
\label{prpo:gsd_lower_bound}
$\alpha$-Geodesical skew divergence satisfies the following inequality for all $\alpha\in\mathbb{R},\lambda\in[0,1]$.

\begin{equation}
    D_{GS}^{(\alpha,\lambda)}[p\|q] \geq \int_\mathcal{X} p\ln \frac{p}{\max\{p, q\}}d\mu.
\end{equation}

\end{Proposition}
\begin{proof}
It follows from the definition of the inverse monotonicity of $f$-interpolation~\eqref{eq:f_interpolation_inverse} and Lemma~\ref{lem:gsd_alpha_minus_infty}.
\end{proof}

\begin{Proposition}{(Upper bound of the $\alpha$-geodesical skew divergence)}
\label{prpo:gsd_upper_bound}
$\alpha$-Geodesical skew divergence satisfies the following inequality for all $\alpha\in\mathbb{R},\lambda\in[0,1]$.

\begin{equation}
    D_{GS}^{(\alpha,\lambda)}[p\|q] \leq \int_\mathcal{X} p\ln\frac{p}{\min\{p, q\}}d\mu.
\end{equation}

\end{Proposition}
\begin{proof}
It follows from the definition of the  $f$-interpolation~\eqref{eq:f_interpolation_inverse} and Lemma~\ref{lem:gsd_alpha_infty}.
\end{proof}

\begin{Theorem}{(Strong convexity of the $\alpha$-geodesical skew divergence)}
$\alpha$-Geodesical skew divergence $D_{GS}^{(\alpha,\lambda)}[p\|q]$ is strongly convex in $p$ with respect to the total variation norm.
\end{Theorem}
\begin{proof}
Let $r\coloneqq m_f^{(\alpha,\lambda)}(p,q)$ and $f_j\coloneqq\frac{p_j}{r}\ (j=0,1)$, so that $f_t = \frac{p_t}{r}\ (t\in(0,1))$.
From Taylor's theorem, for $g(x)\coloneqq x\ln x$ and $j=0,1$, it holds that

\begin{align}
    g(f_j) = g(f_t) + g'(f_t)(f_j - f_t) + (f_j - f_t)^2\int^1_0 g^{''}((1-s)f_t + s f_j)(1-s) ds. \nonumber
\end{align}

Let

\begin{align*}
    \delta &\coloneqq (1-t)g(f_0) + t g(f_1) - g(f_t) \\
    &= (1-t)t(f_1-f_0)^2\int^1_0\Biggl(\frac{t}{(1-s)f_t + s f_0} + \frac{1-t}{(1-s)f_t + s f_1} \Biggr)(1-s)ds \\
    &= (1-t)t(f_1-f_0)^2\int^1_0 \Biggl(\frac{t}{f_{u_0}(t,s)} + \frac{1-t}{f_{u_1}(t,s)}\Biggr)(1-s)ds,
\end{align*}

where

\begin{align*}
    u_j(t,s) &\coloneqq (1-s)t + j t, \\
    f_{\mu_j}(t,s) &\coloneqq (1-s)f_t + s f_j.
\end{align*}

Then,

\begin{align*}
    \Delta &\coloneqq (1-t)H(p_0) + tH(p_1) - H(p_t) \\
    &= \int\delta dr \\
    &= (1-t)t\int^1_0 (1-s)ds\Big[t I(u_0(t,s)) + (1-t)I(u_1(t,s))\Big],
\end{align*}

where

\begin{align*}
    \|p_1 - p_0\| &\coloneqq \int |dp_1 - dp_0|d\mu, \\
    H(p) &\coloneqq D_{GS}^{(\alpha,\lambda)}[p\|r] = \int p\ln\frac{p}{r}d\mu, \\
    I(u) &\coloneqq \int \frac{(f_1-f_0)^2}{f_u} dr.
\end{align*}

Now, it is suffice to prove that $\Delta \geq \frac{t(1-t)}{2}\|p_1 - p_0 \|^2$.
For all $u\in(0,1)$, it is seen that $p_1$ is absolutely continuous with respect to $p_u$.
Let $g_u\coloneqq \frac{p_1}{p_u} = \frac{f_1}{f_u}$. One obtains

\begin{align*}
    I(u) &= \frac{1}{(1-u)^2}\int\frac{(f_1-f_u)^2}{f_u}dr \\
    &= \frac{1}{(1-u)^2}\int (g_u-1)^2 dp_u \\
    &\geq \frac{1}{(1-u)^2}\Biggl(\int |g_u - 1|dp_u \Biggr)^2 \\
    &= \frac{1}{(1-u)^2}\|p_1 - p_u\|^2 = \|p_1 - p_0\|^2,
\end{align*}

and hence, for $j=0,1$,

\begin{equation*}
    \Delta \geq \frac{t(1-t)}{2}\|p_1 - p_0 \|^2.
\end{equation*}
\end{proof}

\section{Natural \texorpdfstring{$\alpha$}--Geodesical Skew Divergence for Exponential Family}
In this section, the exponential  family is considered in which probability density function is given by

\begin{equation}
    p(\bm{x};\bm{\theta}) = \exp\Big\{\bm{\theta}\cdot\bm{x} +k(\bm{x}) - \psi(\bm{\theta})\Big\},
\end{equation}

where $\bm{x}$ is a random variable. In the above equation, $\bm{\theta}=(\theta^1,\dots,\theta^n)$ is an $n$-dimensional vector parameter to specify distribution, $k(\bm{x})$ is a function of $\bm{x}$ and $\psi$ corresponds to the normalization factor.

In skew divergence, the probability distribution of the target is a weighted average of the two distributions.
This implicitly assumes that interpolation of the two probability distributions is properly given by linear interpolation.
Here, in the exponential family, the  interpolation between natural parameters rather than interpolation between probability distributions themselves is considered. Namely, the geodesic connecting two distributions $p(\bm{x};\bm{\theta}_p)$ and $q(\bm{x}; \bm{\theta}_q)$ on the $\bm{\theta}$-coordinate system is considered:

\begin{equation}
    \bm{\theta}(\lambda) = (1-\lambda)\bm{\theta}_p + \lambda\bm{\theta}_q,
\end{equation}

where $\lambda\in[0,1]$ is the parameter.
The probability distributions on the geodesic $\bm{\theta}(\lambda)$ are

\begin{align}
    p(\bm{x};\lambda) &= p(\bm{x}; \bm{\theta}(\lambda)) \nonumber \\
    &= \exp\Big\{\lambda(\bm{\theta}_q - \bm{\theta}_p)\cdot\bm{x} + \bm{\theta}_p\cdot \bm{x} - \psi(\lambda)\Big\}.
\end{align}

Hence, a geodesic itself is a one-dimensional exponential family, where $\lambda$ is the natural parameter.
A geodesic consists of a linear interpolation of the two distributions in the logarithmic scale because

\begin{equation}
    \ln p(\bm{x};\lambda) = (1-\lambda)\ln{p(\bm{x};\bm{\theta}_p)} + \lambda \ln p(\bm{x};\bm{\theta}_q) - \psi(\lambda).
\end{equation}

This corresponds to the case $\alpha=1$ on the $f$-interpolation with normalization factor $c(\lambda)=\exp{\{-\psi(\lambda)\}}$,

\begin{equation}
    p(\bm{x};\bm{\theta}(\lambda)) = m_f^{(\lambda, 1)}(p(\bm{x};\bm{\theta}_p), p(\bm{x};\bm{\theta}_q)).
\end{equation}

This induces the natural geodesic skew divergence with $\alpha=1$ as

\begin{align*}
    D_{GS}^{(1,\lambda)}[p \| q] &= \int_\mathcal{X} p\ln \Biggr(\frac{p}{m_f^{(\lambda, 1)}(p,q)}\Biggl) d\mu \\
    &= \int_\mathcal{X} p\ln p - p\ln \Big(m_f^{(\lambda,1)}(p,q)\Big) d\mu \\
    &= \int_\mathcal{X} p\ln p - p\ln \Big(\exp\{(1-\lambda)\ln p + \lambda\ln q\}\Big) d\mu \\
    &= \int_\mathcal{X} \Big( p\ln p - (1-\lambda)p\ln p - \lambda p\ln q \Big) d\mu \\
    &= \int_\mathcal{X} \Big(\lambda p\ln p - \lambda p\ln q \Big) d\mu \\
    &= \lambda\int_\mathcal{X} p\ln\frac{p}{q} d\mu \\
    &= \lambda D_{KL}[p\|q],
\end{align*}

and this is equal to the scaled KL divergence.

\begin{figure}
    \centering
    \includegraphics[scale=0.4]{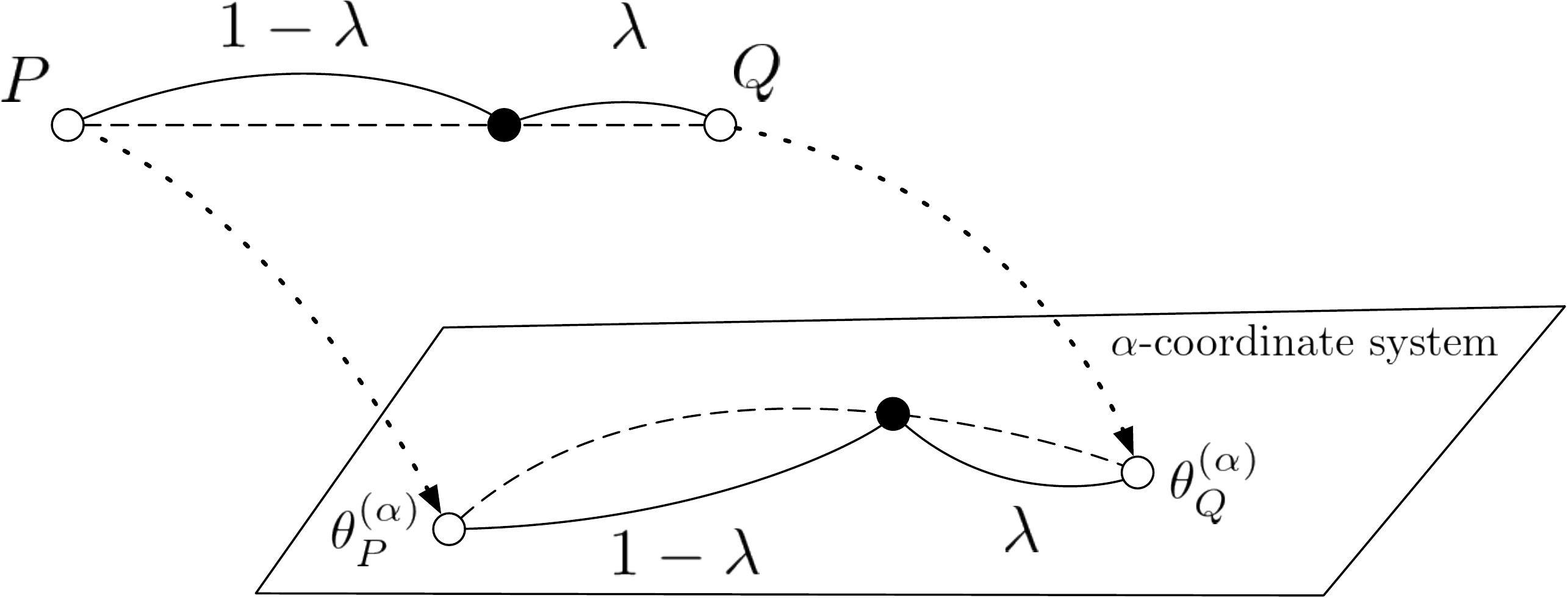}
    \caption{The geodesic between two probability distributions on the $\alpha$-coordinate system.}
    \label{fig:alpha_gsd}
\end{figure}

More generally, let $\theta_P^{(\alpha)}$ and $\theta_Q^{(\alpha)}$ be the parameter representations on the $\alpha$-coordinate system of probability distributions $P$ and $Q$. Then, the geodesics between them are represented as in Figure~\ref{fig:alpha_gsd} and it induces the $\alpha$-geodesical skew divergence.

\section{Function Space associated with the  \texorpdfstring{$\alpha$}--Geodesical Skew Divergence}

To discuss the functional nature of the $\alpha$-geodesical skew divergence in more depth, the function space it constitutes is considered.
For an $\alpha$-geodesical skew divergence $f^{(\alpha,\lambda)}_{q}(p) = D_{GS}^{(\alpha,\lambda)}[p\|q]$ with one side of the distribution fixed, let the entire set be

\begin{equation}
    \mathcal{F}_{q} = \Big\{f^{(\alpha,\lambda)}_{q}\mid\alpha\in\mathbb{R}, \lambda\in[0,1] \Big\}.
\end{equation}

For $f^{(\alpha,\lambda)}_{q} \in \mathcal{F}_q$, its semi-norm is defined by
\begin{equation}
    \Big\|f^{(\alpha,\lambda)}_{q}\Big\|_p\coloneqq \int_\mathcal{X}\Big(\Big|f^{(\alpha,\lambda)}_{q}\Big|^p d\mu\Big)^{\frac{1}{p}}. \label{eq:lp_norm}
\end{equation}
By defining addition and scalar multiplication for $f_q^{(\alpha,\lambda)},g_q^{(\alpha,\lambda)}\in\mathcal{F}_q$, $c\in\mathbb{R}$ as follows, $\mathcal{F}_q$ becomes a semi-norm vector space:

\begin{align}
   (f_q^{(\alpha,\lambda)} + g_q^{(\alpha,\lambda)})(u) &\coloneqq f_q^{(\alpha,\lambda)}(u) + g_q^{(\alpha,\lambda)}(u) = D_{GS}^{(\alpha,\lambda)}[u\|q] + D_{GS}^{(\alpha',\lambda')}[u\|q], \\
   (cf)(u) &\coloneqq cf_q^{(\alpha,\lambda)}(u) = c\cdot D_{GS}^{(\alpha,\lambda)}[u\|q].
\end{align}

\begin{Theorem}
Let $\mathcal{N}$ be the kernel of $\|\cdot\|_p$ as follows:

\begin{equation}
    \mathcal{N} \coloneqq ker(\|\cdot \|_p) = \Big\{f_q^{(\alpha,\lambda)}\mid f_q^{(\alpha,\lambda)} = 0 \Big\}.
\end{equation}

Then the quotient space $\mathcal{V}\coloneqq(\mathcal{F}_{q},\|\cdot\|_p)/N$ is a Banach space.
\end{Theorem}
\begin{proof}
It is sufficient to prove that $f^{(\alpha,\lambda)}_{q}$ is integrable to the power of $p$ and that $\mathcal{V}$ is complete.
From Proposition~\ref{prpo:gsd_upper_bound}, the $\alpha$-geodesical skew divergence is bounded from above for all $\alpha\in\mathbb{R}$ and $\lambda\in[0,1]$.
Since $f_q^{(\alpha,\lambda)}$ is continuous, we know that it is $p$-power integrable.

Let $\{f_n\}$ be a Cauchy sequence of $\mathcal{V}$:

\begin{equation*}
    \lim_{n,m\to\infty}\|f_n - f_m\|_p = 0.
\end{equation*}

Since $n(k),\ k=1,2,\dots,$ can be taken to be monotonically increasing and

\begin{equation*}
    \|f_n - f_{n(k)}\|_p < 2^{-k}
\end{equation*}

with respect to $n>n(k)$, let

\begin{equation*}
    \|f_{n(k+1)} - f_{n(k)}\|_p < 2^{-k}.
\end{equation*}

If $g_n = |f_{n(1)}| + \sum^{n-1}_{j=1}|f_{n(j+1)}-f_{n(j)}|\in\mathcal{V}$, it is non-negatively monotonically increasing at each point, and from the subadditivity of the norm, $\|g_n\|_p\leq\|f_{n(1)}\|_p + \sum^{n-1}_{j=1}2^{-j}$.
From the monotonic convergence theorem, we have

\begin{equation*}
    \Big\|\lim_{n\to\infty} g_n\Big\|_p = \lim_{n\to\infty}\|g_n\|_p \leq \|f_{n(1)}\|_p + 1 < \infty.
\end{equation*}

That is, $\lim_{n\to\infty}g_n$ exists almost everywhere, and $\lim_{n\to\infty}g_n\in\mathcal{V}$.
From $\lim_{n\to\infty}g_n <\infty$, we have

\begin{equation*}
    f_{n(1)} + \sum^{n-1}_{j=1}(f_{n(j+1)}-f_{n(j)}) = \lim_{n\to\infty} f_{n(1)}
\end{equation*}

converges absolutely almost everywhere to $|\lim_{n\to\infty} f_{n(n)}|\leq \lim_{n\to\infty}g_n, a.e.$.
That is, $\lim_{n\to\infty} f_{n(n)}\in\mathcal{V}$.
Then

\begin{equation*}
    \Big|\lim_{n\to\infty}f_n - f_{n(n)} \Big| \leq \lim_{n\to\infty}g_n
\end{equation*}

and from the superior convergence theorem, we can obtain

\begin{equation*}
    \lim_{n\to\infty}\Big\|\lim_{n\to\infty}f_n - f_{n(n)} \Big\|_p = 0
\end{equation*}

We have now confirmed the completeness of $\mathcal{V}$.
\end{proof}

\begin{figure}[t]
    \centering
    \includegraphics[scale=0.38]{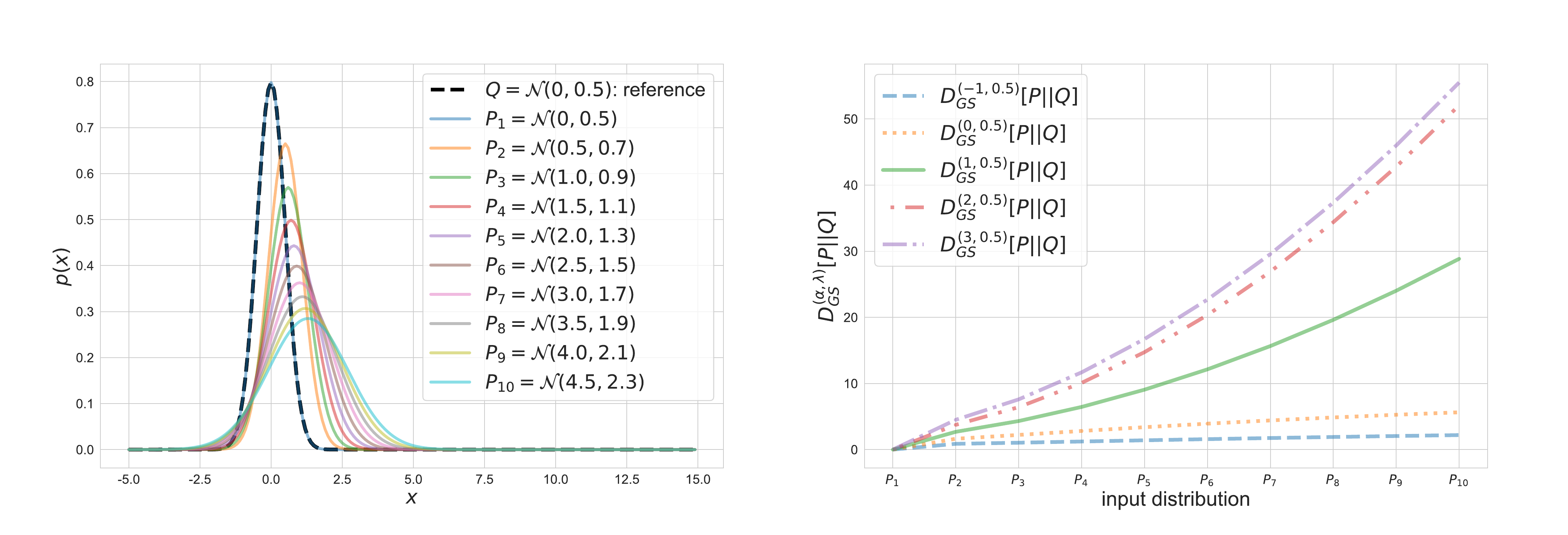}
    \caption{$\alpha$-geodesical skew divergence between two normal distributions. The reference distribution is $Q=\mathcal{N}(0, 0.5)$.
    For $P_1,P_2,\dots,P_j,\ (j=1,2,\dots,10)$, let their mean and variance be $\mu_j$ and $\sigma^2_j$, respectively, where $\mu_{j+1} - \mu_j = 0.5$ and $\sigma^2_{j+1} - \sigma^2_j = 0.2$.}
    \label{fig:gsd_gaussian}
\end{figure}

\begin{figure}[t]
    \centering
    \includegraphics[scale=0.4]{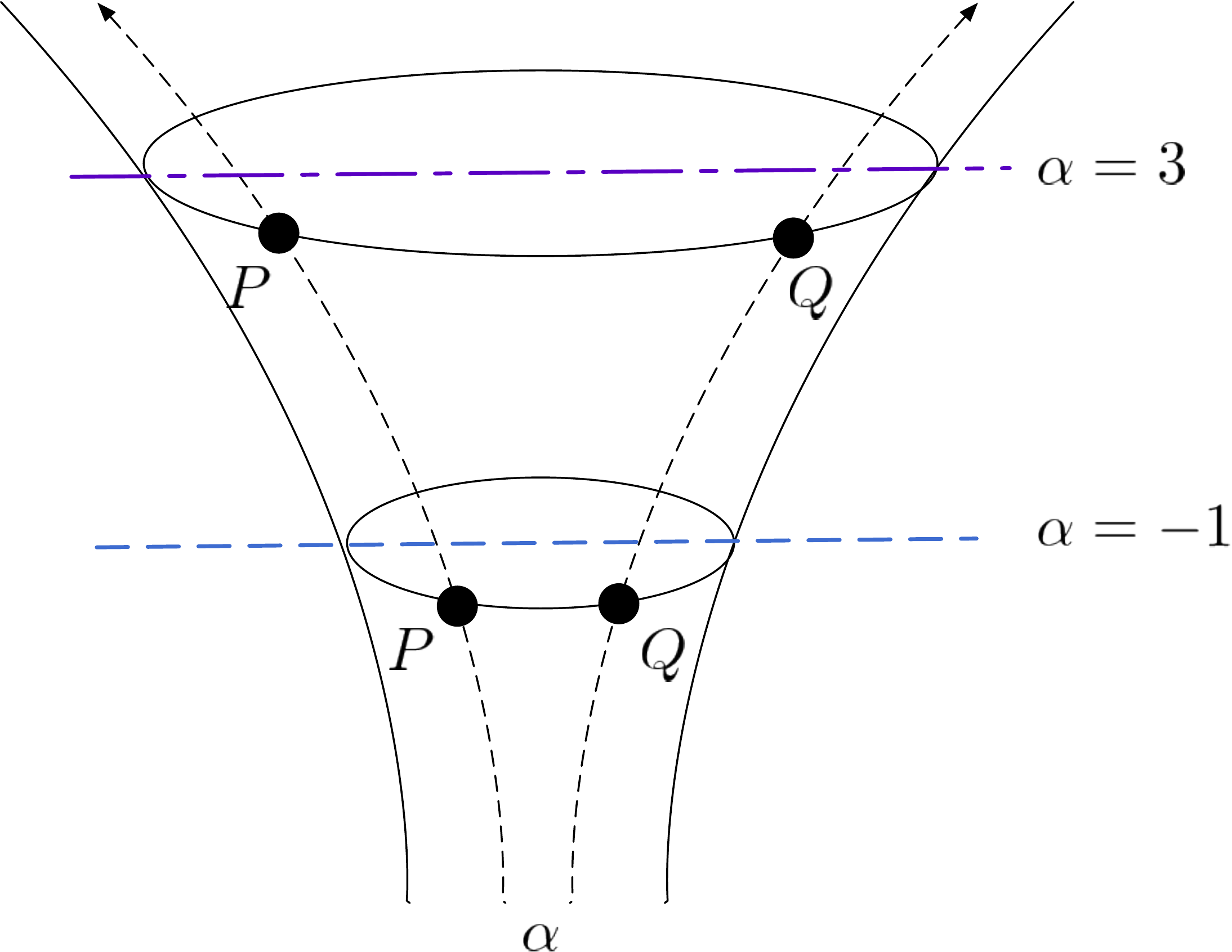}
    \caption{Coordinate system of $\mathcal{F}_q$ or $\mathcal{F}_+$.
    Such a coordinate system is not Euclidean.}
    \label{fig:coordinate_system_banach}
\end{figure}

\begin{Corollary}
Let

\begin{equation}
    \mathcal{F}_+ = \Big\{f_q^{(\alpha,\lambda)}\mid\alpha\in\mathbb{R}, \lambda\in(0,1], q\in\mathcal{P} \Big\}.
\end{equation}

Then the space $\mathcal{V}_+ \coloneqq (\mathcal{F}_+, \|\cdot\|_p)$ is a Banach space.
\end{Corollary}
\begin{proof}
If we restrict $\lambda\in(0,1]$, $D_{GS}^{(\alpha, \lambda)}[u\|q]=0$ if and only if $u=q$.
Then, $\mathcal{V}_+$ has the unique identity element, and then $\mathcal{V}_+$ is a complete norm space.
\end{proof}

Consider the socond argument $Q$ of $D^{(\alpha,\lambda)}_{GS}(P||Q)$ is fixed, which is referred to as 
the reference distribution.
Figure~\ref{fig:gsd_gaussian} shows values of the $\alpha$-geodesical skew divergence for a fixed reference $Q$, where both $P$ and $Q$ are restricted to be Gaussian.
In this figure, the reference distribution is $\mathcal{N}(0, 0.5)$ and the parameters of input distributions are varied in $\mu\in[0, 4.5]$ and  $\sigma^2\in[0.5, 2.3]$.
From this figure, one can see that larger value of $\alpha$ emphasizes the discrepancy between distributions $P$ and $Q$.
Figure~\ref{fig:coordinate_system_banach} illustrates a coordinate system associated with the $\alpha$-geodesical skew divergence for different $\alpha$. As seen from the figure, for the same pair of distributions $P$ and $Q$, the value of divergence with $\alpha=3$ is larger than that with $\alpha=-1$.

\section{Conclusions and Discussion}
In this paper, a new family of divergence is proposed to address the computational difficulty of KL-divergence.
The proposed $\alpha$-geodesical skew divergence is a natural derivation from the concept of $\alpha$-geodesics in information geometry and generalizes many existing divergences.
Furthermore, $\alpha$-geodesical skew divergence leads to several applications.
For example, the new divergence can be applied to the annealed importance sampling by the same analogy as in previous studies using q-paths~\cite{Brekelmans2020AnnealedIS}.
It could also be applied to linguistics, a field in which skew divergence is originally used~\cite{lee1999measures}.

\section*{Acknowledgement}
The authors express special thanks to the editor and reviewers whose comments led to valuable improvements of the manuscript.
Part of this work is supported by JPSJ (KAKENHI) grant number JP17H01793, JST CREST Grant No. JPMJCR2015 and NEDO JPNP18002.

\bibliographystyle{unsrt}  
\bibliography{preprint}  

\begin{thebibliography}{10}

\bibitem{deza2009encyclopedia}
Michel~Marie Deza and Elena Deza.
\newblock Encyclopedia of distances.
\newblock In {\em Encyclopedia of distances}, pages 1--583. Springer, 2009.

\bibitem{basseville2013divergence}
Mich{\'e}le Basseville.
\newblock Divergence measures for statistical data processing—an annotated
  bibliography.
\newblock {\em Signal Processing}, 93(4):621--633, 2013.

\bibitem{kullback1951information}
Solomon Kullback and Richard~A Leibler.
\newblock On information and sufficiency.
\newblock {\em The annals of mathematical statistics}, 22(1):79--86, 1951.

\bibitem{sakamoto1986akaike}
Yosiyuki Sakamoto, Makio Ishiguro, and Genshiro Kitagawa.
\newblock Akaike information criterion statistics.
\newblock {\em Dordrecht, The Netherlands: D. Reidel}, 81(10.5555):26853, 1986.

\bibitem{goldberger2003efficient}
Jacob Goldberger, Shiri Gordon, Hayit Greenspan, et~al.
\newblock An efficient image similarity measure based on approximations of
  kl-divergence between two gaussian mixtures.
\newblock In {\em ICCV}, volume~3, pages 487--493, 2003.

\bibitem{yu2013kl}
Dong Yu, Kaisheng Yao, Hang Su, Gang Li, and Frank Seide.
\newblock Kl-divergence regularized deep neural network adaptation for improved
  large vocabulary speech recognition.
\newblock In {\em 2013 IEEE International Conference on Acoustics, Speech and
  Signal Processing}, pages 7893--7897. IEEE, 2013.

\bibitem{solanki2006provably}
Kaushal Solanki, Kenneth Sullivan, Upamanyu Madhow, BS~Manjunath, and Shivkumar
  Chandrasekaran.
\newblock Provably secure steganography: Achieving zero kl divergence using
  statistical restoration.
\newblock In {\em 2006 International Conference on Image Processing}, pages
  125--128. IEEE, 2006.

\bibitem{lin1991divergence}
Jianhua Lin.
\newblock Divergence measures based on the shannon entropy.
\newblock {\em IEEE Transactions on Information theory}, 37(1):145--151, 1991.

\bibitem{menendez1997jensen}
ML~Men{\'e}ndez, JA~Pardo, L~Pardo, and MC~Pardo.
\newblock The jensen-shannon divergence.
\newblock {\em Journal of the Franklin Institute}, 334(2):307--318, 1997.

\bibitem{nielsen2019jensen}
Frank Nielsen.
\newblock On the jensen--shannon symmetrization of distances relying on
  abstract means.
\newblock {\em Entropy}, 21(5):485, 2019.

\bibitem{jeffreys1946invariant}
Harold Jeffreys.
\newblock An invariant form for the prior probability in estimation problems.
\newblock {\em Proceedings of the Royal Society of London. Series A.
  Mathematical and Physical Sciences}, 186(1007):453--461, 1946.

\bibitem{chatzisavvas2005information}
K~Ch Chatzisavvas, Ch~C Moustakidis, and CP~Panos.
\newblock Information entropy, information distances, and complexity in atoms.
\newblock {\em The Journal of chemical physics}, 123(17):174111, 2005.

\bibitem{bigi2003using}
Brigitte Bigi.
\newblock Using kullback-leibler distance for text categorization.
\newblock In {\em European conference on information retrieval}, pages
  305--319. Springer, 2003.

\bibitem{wang2006groupwise}
Fei Wang, Baba~C Vemuri, and Anand Rangarajan.
\newblock Groupwise point pattern registration using a novel cdf-based
  jensen-shannon divergence.
\newblock In {\em 2006 IEEE Computer Society Conference on Computer Vision and
  Pattern Recognition (CVPR'06)}, volume~1, pages 1283--1288. IEEE, 2006.

\bibitem{nishii2006image}
Ryuei Nishii and Shinto Eguchi.
\newblock Image classification based on markov random field models with
  jeffreys divergence.
\newblock {\em Journal of multivariate analysis}, 97(9):1997--2008, 2006.

\bibitem{bayarri2008generalization}
MJ~Bayarri and G~Garc{\'\i}a-Donato.
\newblock Generalization of jeffreys divergence-based priors for bayesian
  hypothesis testing.
\newblock {\em Journal of the Royal Statistical Society: Series B (Statistical
  Methodology)}, 70(5):981--1003, 2008.

\bibitem{nielsen2013jeffreys}
Frank Nielsen.
\newblock Jeffreys centroids: A closed-form expression for positive histograms
  and a guaranteed tight approximation for frequency histograms.
\newblock {\em IEEE Signal Processing Letters}, 20(7):657--660, 2013.

\bibitem{nielsen2020generalization}
Frank Nielsen.
\newblock On a generalization of the jensen--shannon divergence and the
  jensen--shannon centroid.
\newblock {\em Entropy}, 22(2):221, 2020.

\bibitem{lee1999measures}
Lillian Lee.
\newblock Measures of distributional similarity.
\newblock In {\em Proceedings of the 37th annual meeting of the Association for
  Computational Linguistics on Computational Linguistics}, pages 25--32, 1999.

\bibitem{lee2001effectiveness}
Lillian Lee.
\newblock On the effectiveness of the skew divergence for statistical language
  analysis.
\newblock In {\em AISTATS}. Citeseer, 2001.

\bibitem{xiao2019dual}
Fengshun Xiao, Yingting Wu, Hai Zhao, Rui Wang, and Shu Jiang.
\newblock Dual skew divergence loss for neural machine translation.
\newblock {\em arXiv preprint arXiv:1908.08399}, 2019.

\bibitem{carvalho2014skew}
Bruno~M Carvalho, Edgar Gardu{\~n}o, and Ira{\c{c}}{\'u}~O Santos.
\newblock Skew divergence-based fuzzy segmentation of rock samples.
\newblock In {\em Journal of Physics: Conference Series}, volume 490, page
  012010. IOP Publishing, 2014.

\bibitem{revathi2014cotton}
P~Revathi and M~Hemalatha.
\newblock Cotton leaf spot diseases detection utilizing feature selection with
  skew divergence method.
\newblock {\em International Journal of scientific engineering and technology},
  3(1):22--30, 2014.

\bibitem{ahmed2011network}
Nesreen Ahmed, Jennifer Neville, and Ramana~Rao Kompella.
\newblock Network sampling via edge-based node selection with graph induction.
\newblock 2011.

\bibitem{hughes2007lexical}
Thad Hughes and Daniel Ramage.
\newblock Lexical semantic relatedness with random graph walks.
\newblock In {\em Proceedings of the 2007 joint conference on empirical methods
  in natural language processing and computational natural language learning
  (EMNLP-CoNLL)}, pages 581--589, 2007.

\bibitem{audenaert2014quantum}
Koenraad~MR Audenaert.
\newblock Quantum skew divergence.
\newblock {\em Journal of Mathematical Physics}, 55(11):112202, 2014.

\bibitem{hardy1952inequalities}
Godfrey~Harold Hardy, John~Edensor Littlewood, and George P{\'o}lya.
\newblock {\em Inequalities. By GH Hardy, JE Littlewood, G. P{\'o}lya..}
\newblock University Press, 1952.

\bibitem{Amari2016-pi}
Shun-Ichi Amari.
\newblock {\em Information Geometry and Its Applications}.
\newblock Springer, 2 2016.

\bibitem{kolmogorov1930notion}
Andrey~Nikolaevich Kolmogorov and Guido Castelnuovo.
\newblock {\em Sur la notion de la moyenne}.
\newblock G. Bardi, tip. della R. Accad. dei Lincei, 1930.

\bibitem{nagumo1930klasse}
Mitio Nagumo.
\newblock {\"U}ber eine klasse der mittelwerte.
\newblock In {\em Japanese journal of mathematics: transactions and abstracts},
  volume~7, pages 71--79. The Mathematical Society of Japan, 1930.

\bibitem{nielsen2014generalized}
Frank Nielsen.
\newblock Generalized bhattacharyya and chernoff upper bounds on bayes error
  using quasi-arithmetic means.
\newblock {\em Pattern Recognition Letters}, 42:25--34, 2014.

\bibitem{amari2012differential}
Shun-ichi Amari.
\newblock {\em Differential-geometrical methods in statistics}, volume~28.
\newblock Springer Science \& Business Media, 2012.

\bibitem{Amari1985-mi}
Shunichi Amari.
\newblock Differential-geometrical methods in statistics.
\newblock {\em Lecture Notes on Statistics}, 28:1, 1985.

\bibitem{Amari2009-fz}
S~Amari.
\newblock $\alpha$ -divergence is unique, belonging to both {$f$-Divergence}
  and bregman divergence classes.
\newblock {\em IEEE Trans. Inf. Theory}, 55(11):4925--4931, November 2009.

\bibitem{Ay2017}
Nihat Ay, J\"{u}rgen Jost, H{\^{o}}ng~V{\^{a}}n L{\^{e}}, and Lorenz
  Schwachh\"{o}fer.
\newblock {\em Information Geometry}.
\newblock Springer International Publishing, 2017.

\bibitem{Morozova1991}
E.~A. Morozova and N.~N. Chentsov.
\newblock Markov invariant geometry on manifolds of states.
\newblock {\em Journal of Soviet Mathematics}, 56(5):2648--2669, October 1991.

\bibitem{Eguchi2015}
Shinto Eguchi and Osamu Komori.
\newblock Path connectedness on a space of probability density functions.
\newblock In {\em Lecture Notes in Computer Science}, pages 615--624. Springer
  International Publishing, 2015.

\bibitem{e23040464}
Frank Nielsen.
\newblock On a variational definition for the jensen-shannon symmetrization of
  distances based on the information radius.
\newblock {\em Entropy}, 23(4), 2021.

\bibitem{nielsen2010family}
Frank Nielsen.
\newblock A family of statistical symmetric divergences based on jensen's
  inequality.
\newblock {\em arXiv preprint arXiv:1009.4004}, 2010.

\bibitem{cover1999elements}
Thomas~M Cover.
\newblock {\em Elements of information theory}.
\newblock John Wiley \& Sons, 1999.

\bibitem{Brekelmans2020AnnealedIS}
Rob Brekelmans, Vaden Masrani, Thang~D. Bui, Frank~D. Wood, A.~Galstyan, G.~V.
  Steeg, and F.~Nielsen.
\newblock Annealed importance sampling with q-paths.
\newblock {\em ArXiv}, abs/2012.07823, 2020.

\end{thebibliography}

\end{document}